%% file: hmm_isit.tex
\documentclass[conference,letterpaper]{IEEEtran}

\usepackage{bbm,amssymb,amsmath,amsfonts,latexsym}
\usepackage{graphicx,color,subfigure}
\usepackage{tikz}
\usepackage{algorithm,algorithmic}
\usepackage{graphics,url}
\usepackage{multirow}
\usepackage{hyperref}
\definecolor{subtler}{rgb}{1,0,0.1}  

\newcommand{\be}{\begin{equation}}
\newcommand{\ee}{\end{equation}}
\newcommand{\ben}{\begin{equation*}}
\newcommand{\een}{\end{equation*}}
\newcommand{\ba}{\begin{eqnarray}}
\newcommand{\ea}{\end{eqnarray}}

\newtheorem{thm}{Theorem}

\newtheorem{lem}[thm]{Lemma}

\newtheorem{defn}[thm]{Definition}
\newtheorem{rem}[thm]{Remark}
\newcommand{\id}[1]{\mathbf{I}_{#1}}
\newcommand{\mnull}{\mathbf{0}}

\newcommand{\Real}{\mathbb R}

\newcommand{\T}{\mathrm{T}}

\newcommand{\ones}[1]{\mathbf{1}_{#1}}
\newcommand{\tones}[1]{\mathbf{\tilde 1}_{#1}}

\def\rank{\mathrm{rank}}

\def\bA{\mathbf{A}}
\def\tA{\mathbf{\tilde A}}
\def\bB{\mathbf{B}}
\def\tB{\mathbf{\tilde B}}
\def\bC{\mathbf{C}}
\def\tC{\mathbf{\tilde C}}
\def\bD{\mathbf{D}}
\def\bE{\mathbf{E}}

\def\bH{\mathbf{H}}

\def\bW{\mathbf{W}}

\def\bX{\mathbf{X}}

\def\bx{\mathbf{x}}

\def\by{\mathbf{y}}

\def\ba{\mathbf{a}}
\def\bb{\mathbf{b}}
\def\bc{\mathbf{c}}

\def\bpi{\boldsymbol{\pi}}
\def\tpi{\boldsymbol{\tilde \pi}}
\def\blambda{\boldsymbol{\lambda}}
\def\tlambda{\boldsymbol{\tilde \lambda}}
\def\blambdam{\blambda_{\mathrm{multi}}}

\def\ta{\mathbf{\tilde a}}
\def\tb{\mathbf{\tilde b}}

\def\supp{\mathrm{supp}}

\def\orow{\otimes^{\mathrm{row}}}
\newcommand{\bigorow}[2]{{\bigotimes_{#1}^{#2}{}^{\mathrm{row}}}\,}
\def\krank{\mathrm{krank}}

\def\bcX{\boldsymbol{\mathcal{X}}}

\def\cX{\mathcal{X}}
\def\cY{\mathcal{Y}}

\def\bcL{\boldsymbol{\mathcal{L}}}
\def\bcM{\boldsymbol{\mathcal{M}}}

\def\Perm{\boldsymbol{\Pi}}
\def\Scale{\boldsymbol{\Lambda}}

\def\techrep{\cite{Tune12HMMTensor}}

\graphicspath{{.}
{Figures/}
}

\begin{document}

\title{Hidden Markov Model Identifiability via Tensors}

\author{Paul~Tune,~Hung~X.~Nguyen~and Matthew~Roughan\\ %
\authorblockA{School of Mathematical Sciences, University of Adelaide,\\ Adelaide, SA, Australia.\\
Email: \{paul.tune, hung.nguyen, matthew.roughan\}@adelaide.edu.au}}

\maketitle

\input{abstract}

\input{intro}

\input{hmm}

\input{tensor}

\input{single_identifiability}

\input{multi_identifiability}

\input{conclusion}

\bibliography{hmm_bib}
\bibliographystyle{abbrv}

\end{document}

%% file: abstract.tex
\begin{abstract}
The prevalence of hidden Markov models (HMMs) in various applications of statistical signal processing and communications is 
a testament to the power and flexibility of the model. In this paper, we link the identifiability problem with tensor decomposition, in
particular, the Canonical Polyadic decomposition. Using recent results in deriving uniqueness conditions for tensor
decomposition, we are able to provide a necessary and sufficient condition for the identification of the 
parameters of discrete time finite alphabet HMMs. This result resolves a long standing open problem 
regarding the derivation of a necessary and sufficient condition for uniquely identifying an HMM. We then further extend recent preliminary work on the identification of HMMs 
with multiple observers by deriving necessary and sufficient conditions for identifiability in this setting. 
\end{abstract}

%% file: intro.tex
\section{Introduction}

The hidden Markov model (HMM) was first introduced in the late 1950s by Blackwell and Koopmans 
\cite{Blackwell57HMM} and generalised later by Baum and Petrie \cite{Baum66HMM}. HMMs have been applied to a variety of domains, such as signal processing, machine learning, communications and many more,
with particular emphasis on the inference of the parameters of the HMM, in particular, the hidden states of the system. 
Typically, an unbiased, or asymptotically unbiased, estimator such as a maximum likelihood estimator, is used to infer 
these states, using algorithms such as the famed Baum-Welch algorithm \cite{Baum70ForwBack}. However, identifiability 
conditions, required to ensure the existence of an unbiased estimator, are generally not well-known, with only a select 
number of works  proposing these conditions \cite{Allman09Latent,Finesso90HMM,Gilbert59HMM}. These conditions are
probabilistic and difficult to verify in practice.

In this paper, we derive an identifiability condition for a stationary discrete time HMM, where the 
observations are the realisations of a probabilistic function. We show a strong connection between the 
identifiability of HMMs and the uniqueness of the Canonical Polyadic (CP) tensor decomposition (see \cite{Kolda09Tensor}).
Specifically through a tensor model called the restricted CP model, we derive a necessary and
sufficient condition by using a result by Kruskal \cite{Kruskal77Decomp}, called the permutation lemma. Our main result 
resolves an open problem regarding the derivation of a necessary and sufficient condition for uniquely identifying an HMM. 
A highlight of our results is that the condition is deterministic, compared to generic (probabilisitic) identifiability results. They 
are also easier to verify compared to previous conditions.

These results are particularly helpful in studying the recently proposed multi-observer HMMs 
\cite{Li00TrainingMulti,Nguyen12MultiHMM}. We 
consider two settings: the \emph{homogeneous} setting where all observers possess the same observation matrix, and the 
\emph{heterogeneous} setting, where at least two observers have distinct observation matrices. Surprisingly, the condition
for identifiability in the homogeneous setting is equivalent to having just a single observer of the HMM. Thus, if the HMM
cannot be identified with a single observer, no additional number of independent homogeneous observers can hope to 
identify the hidden states. The heterogeneous setting is shown to provide a significant advantage over the homogeneous 
setting, due to sufficient variability of the observations, contributed by different viewpoints of the independent observers.

The rest of this section introduces the notation used throughout the paper. Section \ref{sec:hmm} formulates our problem and defines the HMM. Section \ref{sec:tensor} provides an overview of
the CP decomposition and some important results in tensor decomposition that we will invoke when proving our results. Section 
\ref{sec:single} derives the identifiability condition of HMMs with only one observer present. 
Section \ref{sec:multi} further extends our framework to the multi-observer setting. Finally, we conclude and outline some 
future work in Section \ref{sec:conclusion}. We defer proofs and more details to our technical report \techrep.

Some notation and definitions are in order. All vectors and matrices are represented with lower and upper case boldface 
fonts respectively. Sets are represented with calligraphic font. Random variables are represented with italic fonts while 
their realisation is represented by lower case italic fonts. Tensors are represented by upper case, calligraphic boldface 
fonts. Let $\id{n}$ be the identity matrix of size $n \times n$, while $\ones{n}$ denotes a column vector of $n$ ones.
$\supp(\bx)$ denotes the support of vector $\bx$.

The Kronecker, or tensor, product between two matrices $\bA \in \Real^{m \times n}$ and $\bB \in \Real^{p \times q}$ is 
denoted by $\bA \otimes \bB \in \Real^{mp \times nq}$. We also define a row-wise tensor product, which we borrow from 
\cite{Allman09Latent}, where with matrices $\bA \in \Real^{m\times n_1}$ and $\bB \in \Real^{m\times n_2}$ with rows 
$\ba_1,\ba_2, \cdots, \ba_m$ and $\bb_1,\bb_2, \cdots, \bb_m$ respectively, the row-wise tensor product is equivalent to
\ben
\bA \orow \bB = 
\begin{bmatrix}
\ba_1 \otimes \bb_1\\
\ba_2 \otimes \bb_2\\
\vdots\\
\ba_m \otimes \bb_m
\end{bmatrix}.
\een
This definition is related to the Khatri-Rao product, which is the column-wise tensor product. Note that for row vectors 
$\ba$ and $\bb$, $\ba \orow \bb = \ba \otimes \bb$. All other notation will be defined on an as needed basis.

%% file: hmm.tex
\section{Hidden Markov Models}
\label{sec:hmm}

Throughout the paper, we only consider the discrete time finite alphabet HMM. Time is organised into regularly spaced discrete 
intervals. Let $\{X_t\}_{t \ge 1}$ be the non-observable 
states of an irreducible, aperiodic Markov chain and $\{Y_t\}_{t \ge 1}$ be the observable states, both at 
time $t$, called an \emph{observation process}. Thus, $\{X_t\}_{t \ge 1}$ constitutes the hidden Markov chain, as it cannot be directly measured. Without loss of
generality, let the alphabets of $X_t$ and $Y_t$ be the sets $\cX = \{1,2,\cdots,q\}$ and $\cY = \{1,2,\cdots,\kappa\}$
respectively. We further assume $q \ge 2$ and $\kappa \ge 2$.

We assume $\{X_t\}_{t \ge 1}$ is stationary
for simplicity of exposition, although our results apply to non-stationary Markov chains as well, with appropriate
modifications. The observations $\{Y_t\}_{t \ge 1}$ are assumed to be i.i.d.~with $Y_t$ only dependent on $X_{t}$. The 
HMM is described by the joint process $(X_t,Y_t) \in \cX \times \cY$ for all $t = 1,2,\cdots,N$, in terms of the
state space model
\begin{align*}
X_{t+1} = f(X_{t}),\ 
Y_t = g(X_{t}),
\end{align*}
with the initial state described by an initial random variable $X_1$, from an initial distribution $\Pr(X_1 = i) = \pi_i$, 
denoted by the $q$--length row vector $\bpi$. The function $f(\cdot)$ is probabilistic, and obeys the $q \times q$ \emph
{transition matrix} $\bA$, where its $(i,j)$-th element is $a_{i,j} := \Pr(X_{t+1} = j\,|\,X_{t} = i)$. The function $g(\cdot)$ may be 
deterministic, but here, we consider it a probabilistic function, with the transition of observation states described by the $q
\times \kappa$ \emph{observation matrix} $\bB$, where the $(i,j)$-th element is $b_{i,j} := \Pr(Y_t = j\,|\,X_{t} = i)$. The 
function $g(\cdot)$ is assumed to be surjective. Since the Markov chain is assumed to be stationary, the observation 
process $\{Y_t\}_{t \ge 1}$ is stationary as well. Also, the observation process $\{Y_t\}_{t \ge 1}$ may not be a Markov chain 
in general, even though the input process is. We are now in a position to formalise HMMs.

\smallskip
\begin{defn}
A discrete time finite alphabet HMM is parameterised by the set $\blambda = \{\bpi; q, \kappa, \bA,\bB\}$:
\begin{itemize}
\item $\bpi$: initial state probabilities, which may or may not be sampled from a stationary distribution,
\item $q$: number of hidden states of $X_t$, i.e.~$|\cX| = q$,
\item $\kappa$: number of observation states of $Y_t$, i.e.~$|\cY|~=~\kappa$,
\item $\bA$: $q \times q$ transition matrix of hidden states $X_t$, and
\item $\bB$: $q \times \kappa$ observation matrix of observation process $Y_t$.
\end{itemize}
\label{defn:hmm}
\end{defn}
\smallskip

One way of measuring the complexity of the HMM is by the number of states required to describe the Markov chain. 
The \emph{order} of an HMM is minimum of $|\cX|$ amongst all representations \cite{Finesso90HMM}.
An HMM is \emph{minimal} if it has a representation such that $|\cX|$ is equal to its order.

An \emph{observation letter} $y_t$ is defined as a single realisation of the observation process at time $t$. 
A \emph{sequence} from time $t_1$ to $t_2$ is defined as a series of consecutive observation letters from time $t_1$ to 
$t_2$. A sequence has \emph{length} $N$ if it consists of $N$ observation letters. 


The joint probability of a particular observed sequence $y_1,y_2,\cdots,y_N$ may be described by
\begin{align}
\nonumber
&P_{\blambda} (Y_1=y_1, Y_2= y_2,\cdots, Y_N = y_N) \\
\nonumber
&=\sum_{\bx \in \cX^N} \pi_{x_1} a_{x_1,x_2} b_{x_1,y_1} a_{x_2,x_3} b_{x_2,y_2}\cdots a_{x_{N-1},x_{N}} b_{x_N,y_N}\\
\label{eq:sequence_prob}
& =\bpi  \bW \bE(y_1)  \bW \bE(y_2) \cdots \bW  \bE(y_N) \ones{q},
\end{align}
where $\bE(k)$ is a $\kappa q \times q$ matrix with the $q \times q$ identity matrix in the $k$-th row partition, and 
\ben
\bW = \bB \orow \bA =
\begin{bmatrix}
\bD_1(\bB)\bA\ \bD_2(\bB) \bA\  \cdots\ \bD_\kappa(\bB)\bA 
\end{bmatrix},
\een
where $\bD_k(\bB)$ denotes the diagonal matrix with the $k$-th column of $\bB$ lying on its diagonal. 


\subsection{Equivalence and identifiability of HMMs}

The observation process $\{Y_t\}_{t\ge 1}$ is assumed to admit a representation of a 
Markov chain with $q$ states. It is possible to construct an HMM with more states that generates the same observation process. Let the process be alternatively parameterised by the set $\tlambda = \{\tpi; \tilde q, \tilde \kappa, 
\tA,\tB\}$. Equivalence of HMMs is defined as follows:

\begin{defn}
Two HMMs with parameterisations $\blambda$ and $\tlambda$ respectively are \emph{equivalent} if and only if for all sequences 
$y_1,y_2,\cdots,y_N$, 
\begin{align*}
&P_{\blambda} (Y_1=y_1, Y_2= y_2,\cdots, Y_N = y_N) \\
&\hspace{1cm}= P_{\tlambda} (\tilde Y_1=y_1, \tilde Y_2= y_2,\cdots, \tilde Y_N = y_N),
\end{align*}
for any integer $N \ge 1$.
\label{def:equivalence}
\end{defn}
  
There are two types of identifiability: \emph{deterministic} and \emph{generic identifiability}. Deterministic 
identifiability implies that an HMM satisfying the condition can always be identified. Generic identifiability means that the 
HMM is identified with probability 1, i.e.~identifiability holds everywhere except for some model parameters that lie in a set of 
Lebesgue measure zero. Allman et al.~\cite{Allman09Latent} defines it as all nonidentifiable parameters of the model lying in a 
proper subvariety. 

Our main result is a condition when an HMM can or cannot be deterministically identified. 

%% file: tensor.tex
\section{A Summary on Tensors}
\label{sec:tensor}

As shown in \eqref{eq:sequence_prob}, the joint probability of a sequence can be expressed 
as matrix multiplication of row tensor products. Identifiability then simply boils down to decomposing the product into
factors via tensor decomposition. 

The tensor is essentially a multidimensional array of numbers, with a general overview found in \cite{Kolda09Tensor}. The 
tensor order is the number of indices required to unambiguously label a component of the tensor, called a \emph{way}.

One particular important decomposition of a tensor is the Canonical Polyadic (CP) decomposition. Let the components of a 
tensor be $\bA \in \Real^{q\times m}$, $\bB \in \Real^{q\times n}$ and $\bC \in \Real^{q\times p}$, written succinctly as 
$\lbrack \bA, \bB, \bC\rbrack$. Then a tensor $\bcX$ constructed from these components is expressed as
\be
\lbrack \bcX; \bA,\bB,\bC\rbrack = \sum_{i=1}^q \ba_i \otimes \bb_i \otimes \bc_i,
\label{eq:cp_tensor}
\ee
where $\ba_i, \bb_i, \bc_i$, $i=1,2\cdots,q$ are the rows of $\bA$, $\bB$ and $\bC$ respectively, essentially a decomposition to $q$ single rank tensors. A tensor is \emph{irreducible} 
with $q$ components if and only if it cannot be decomposed to fewer than $q$ components. A tensor $\bcX$ is \emph{permutation 
and scaling indeterminate} if its components are unique up to a scaling and permutation of rows. Hence, for any 
alternative decomposition of $\bcX$ with components $\lbrack \tA, \tB, \tC \rbrack$, there exists a permutation matrix $\Perm$ and 
nonsingular scaling matrices $\Scale_\bA$, $\Scale_\bB$ and $\Scale_\bC$ where $\Scale_\bA \Scale_\bB \Scale_\bC = \id{q}$, 
such that $\tA = \Perm \Scale_\bA \bA$, $\tB = \Perm \Scale_\bB \bB$ and $\tC = \Perm \Scale_\bC \bC$. Finally, an equivalent 
representation of a tensor is given by its mode \emph{matricisation}. For example, the first mode matricisation of $\bcX$ is $\bA^\T 
 (\bB \orow \bC)$, while its second mode matricisation is $\bB^\T (\bC \orow \bA)$.

Surprisingly, under certain mild conditions, a tensor of order 3 and above can yield a unique CP decomposition, 
up to a scaling and permutation of rows of the components, unlike matrices. A general sufficient condition was proposed by 
Kruskal \cite{Kruskal77Decomp}, with its generalisation in \cite{Sidiropoulos99Tensor}.

Central to Kruskal's and our results is the concept of the \emph{Kruskal rank} defined below:
\begin{defn}
The Kruskal rank of a matrix $\bX$, $\krank(\bX)$ is defined as the largest integer $K$ such that any subset of  
$K$ rows is linearly independent.
\label{defn:kruskal}
\end{defn}
Unlike the rank of a matrix, the Kruskal rank changes when one defines it for columns instead. Here, we stick to the above 
definition for rows, since this is directly relevant to our discussions. 




The cornerstone of Kruskal's result, as pointed out by \cite{Jiang04PermLemma,Stegeman05Perm} is Kruskal's permutation 
lemma, here modified for rows. The lemma is key to our proposed identifiability condition.
\begin{lem}[Permutation lemma]
Given two matrices $\bH$ and $\bar\bH$, both with size $q \times r$, suppose that $\bH$ has no 
identically zero rows, and assume the following implication holds for all column vectors $\bx$:
\begin{align*}
&|\supp(\bar\bH\bx)| \le q-\rank(\bar\bH)+1 \\
&\hspace{1cm} \text{ implies that } |\supp(\bH\bx)| \le |\supp(\bar\bH\bx)|.
\end{align*}
Then, $\bar\bH = \Perm \Scale \bH$, where $\Perm$ is a permutation matrix and $\Scale$ is a nonsingular diagonal scaling 
matrix.
\label{lem:kruskal_perm}
\end{lem}


%% file: single_identifiability.tex
\section{Single observer HMM Setting}
\label{sec:single}

Intuitively, the identifiability of an HMM rests on the number of states $q$ and the number of observation states $\kappa$, 
both having a direct relationship with $\bA$ and $\bB$ respectively. Our reformulation using the properties of tensors allows us to explore this relationship. 
Since the underlying Markov chain is assumed to be irreducible and aperiodic and assuming all alphabets in $\cY$ are not redundant, $\krank(\bA) \ge 1$ and $\krank(\bB) \ge 1$ respectively. 

\subsection{Main result}

Our main result is the following:
\begin{thm}
For an HMM parameterised by $\blambda$ to be unique up to a scaling and permutation of states, it is necessary and 
sufficient that $\krank(\bB\orow \bA) = q$.
\label{thm:unique_identifiability}
\end{thm}

Previous work \cite{Finesso90HMM,Gilbert59HMM} studied the class of \emph{regular} HMMs. An HMM is regular if there 
exists a set of $2q$ sequences whose joint probabilities can be described by a product of two linear subspaces of 
dimension $q$. A regular HMM is permutation and scaling indeterminate. Finesso \cite{Finesso90HMM} provided a simple 
sufficient condition for equivalence between two HMMs. Additionally, the author proved a necessary condition for a regular 
HMM to be equivalent to another HMM. The proof is probabilistic as he showed the set of parameters $\blambda$ of
HMMs almost surely leads to regularity in the Lebesgue measure. 

Our result differs from Finesso's in the sense that we show an interaction between the hidden and the observation states, 
dispensing with assumption of regularity of the HMM and replacing it with a deterministic condition. Furthermore, regular 
HMMs are also minimal, implying $\krank(\bA) = q$ (see \techrep). 
The result shows there is no longer any restriction to regular HMMs, or as coined by Finesso \cite{Finesso90HMM}, a 
\emph{Petrie point} after Petrie's work \cite{Petrie69HMM} on regular HMMs, since the deterministic condition covers all
possible cases. In this sense, the result is the strongest to date on the identifiability, whether deterministic or
generic, of HMMs. 

Theorem \ref{thm:unique_identifiability} is a consequence of the properties of a specific restricted CP tensor, where one 
mode of the tensor is full rank \cite{Jiang04PermLemma}, which we call the \emph{per letter tensor}. Let us consider 
$\bcL$, a three way tensor of dimensions $q \times q \times \kappa$, with component matrices 
$\lbrack \bA, \id{q}, \bB \rbrack$. For each element of 
$\bcL$,
\begin{align*}
\bcL_{i,j,k} &:= P_{\blambda}(X_{t+1} = i\,|\,X_{t} = j) \cdot P_{\blambda}(Y_t = k\,|\,X_t = j)\\
&= P_{\blambda}(Y_t = k,X_{t+1} = i\,|\,X_{t} = j).
\end{align*}
Then, each slice of the third mode $\bcL_k := \bcL_{\cdot,\cdot,k} = \id{q} \bD_k(\bB) \bA$, $k=1,2,\cdots,\kappa$ is the 
per observation letter and state probability of the set $\cY$ arranged in ascending order\footnote{This is one way of labelling the letters, as labelling is non-unique.}. A key observation is the 
equivalence of $\bcL$ and $\id{q} (\bB \orow \bA)$, its second mode matricisation. 

We now prove a necessary and sufficient condition on the uniqueness of the decomposition of $\bcL$. 
%
\begin{lem}
The per letter tensor $\bcL$ is unique up to a permutation and scaling of rows if and only if $\krank(\bB \orow \bA) = q$.
\label{lem:single_letter}
\end{lem}
\begin{IEEEproof}
Crucial to our argument is the central claim that it is necessary and sufficient that none of the non-trivial linear 
combinations of rows of $\bB \orow \bA$ is expressible by a tensor product of two row vectors, that is, $\krank(\bB \orow 
\bA) = q$. Necessity is proven by contradiction. We borrow a counterexample from \cite{Jiang04PermLemma}. If 
the first two rows can be expressed as a vector $\bb_1 \otimes\ba_1 + \bb_2 \otimes \ba_2 = \tb_1 \otimes \ta_1$,
an alternative decomposition of $\bcL$ is as follows:
\begin{align}
\nonumber
\id{q} (\bB \orow \bA) &= 
\begin{bmatrix}
1 & 0 & \mathbf{0}\\
-1 & 1 & \mathbf{0}\\
\mathbf{0} & \mathbf{0} & \id{q-2}
\end{bmatrix}^\T
\begin{bmatrix}
\tb_1 \otimes \ta_1\\
\bb_2 \otimes \ba_2\\
\vdots\\
\bb_q \otimes \ba_q
\end{bmatrix}\\
\label{eq:alternative_decomp}
&= \bC^\T (\tB \orow \tA).
\end{align}
As $\bC^\T$ is not a permutation or scaling, there is no unique decomposition for $\bcL$.

For sufficiency, we only need to verify $|\supp(\bx)| = |\supp(\id{q}\bx)| \le |\supp(\tC\bx)|$ for all $|\supp(\tC\bx)| = 1$ (since $\bx 
=\mnull$ is the 
only zero support vector) for some $\tC$, a component of an alternative decomposition of $\bcL$, to satisfy Lemma 
\ref{lem:kruskal_perm}. With an alternative 
decomposition $\id{q} (\bB \orow \bA) = \tC^\T (\tB \orow \tA)$, then $\forall \bx$,
\ben
\bx^\T (\bB \orow \bA) = \bx^\T\tC^\T  (\tB \orow \tA).
\een
Consider $\bx$ with $|\supp(\tC \bx)| = 1$. Then, $\bx^\T\tC^\T  (\tB \orow \tA)$ is just a scaled tensor product of one row of $\tA$ 
and the corresponding row of $\tB$, by the above equation. If $|\supp(\bx)| > 1$, then more than one row of $\bB \orow \bA$ is needed to 
represent a row of $\tB \orow \tA$. This means a row of $\tB \orow \tA$ is not just a scaling and permutation, so $|\supp(\bx)| \le 1$ must hold. Kruskal's permutation lemma (Lemma \ref{lem:kruskal_perm})
then implies $\id{q}$ and $\tC$ are equivalent up to a permutation and scaling of rows.
Putting these arguments together implies the components of $\bcL$, i.e.~$\bA$, $\bB$ and $\id{q}$ are all unique up to a 
permutation and scaling of rows, and the result follows.
\end{IEEEproof}

The above implies a sufficient condition. For necessity, suppose $\krank(\bB \orow \bA) < q$, but the HMM is identifiable. However, one can construct $\bpi = \tpi \tC^\T,\ \ones{q} = (\tC^\T)^{-1}\tones{q},(\bB \orow \bA) \bE(k)  =  (\tB \orow \tA) (\bE(k)\tC^\T), \forall k,$
with $\tC$ as above in the proof of Lemma \ref{lem:single_letter}, resulting in a contradiction. Thus, $\krank(\bB \orow \bA) = q$ 
is a necessary and sufficient condition for identifiability.



Evaluating the Kruskal rank of $\bB\orow \bA$ is computationally difficult as it requires checking over all possible 
combinations of rows of a matrix. The computational complexity worsens as $q$ and $\kappa$ become large. The 
conditions above may be weakened using the concept of coherence \cite{Donoho03Optimal}, found in compressed 
sensing and dictionary learning literature to derive polynomial time algorithms for verifying HMM identifiability. Further 
details are found in \techrep.

\begin{rem}
The results may be extended to non-stationary Markov chains. In this case, let $\bA(t)$ and $\bB(t)$ be the time-heterogeneous
transition and observation matrices respectively. Then, Theorem \ref{thm:unique_identifiability} may be modified, essentially
asserting that for each $t=1,2,\cdots,N$, the conditions $\krank(\bB(t) \orow \bA(t)) = q$ must hold for the non-stationary HMM to 
be permutation and scaling indeterminate. 
\label{rem:nonstationary_uniqueness}
\end{rem}

%% file: multi_identifiability.tex
\section{Multi-observer HMM Setting}
\label{sec:multi}

The above results prove useful in the study of multi-observer HMMs, which have applications in machine learning 
\cite{Li00TrainingMulti}, and detecting attacks on Internet Service providers \cite{Nguyen12MultiHMM}. We shall 
study identifiability in this setting, but first, we need to lay the foundations for the multi-observer case.

In the multi-observer setting, there are $m \ge 2$ observers of an underlying Markov chain. The observers are assumed 
independent to each other, since dependence would weaken the information content of their observations. The underlying 
irreducible, aperiodic Markov chain $\{X_t\}_{t \ge 1}$ is being observed by all $m$ observers, with each $X_t \in \cX$, starting 
from initial state probabilities $\bpi$, drawn from a stationary distribution. Each observer $j$ may have a different
perspective of the chain, denoted by the processes $\{Y^{(j)}_t\}_{t \ge 1}$, with each $Y^{(j)} \in
\cY^{(j)}$ $\forall j=1,2,\cdots,m$. Without loss of generality, let $\cX = \{1,2,\cdots,q\}$ and for each $j$, $\cY^{(j)} = \{1,2,\cdots,
\kappa_j\}$.

While the transition matrix of the hidden states $\bA$ remains the same for all observers, their associated observation 
matrices may differ. If at least two observation matrices are distinct, i.e.~there exists  indices $\ell$ and $\ell'$ such that 
$\bB^{(\ell)} \ne \bB^{(\ell')}$, the set of independent observers are called \emph{heterogeneous} observers, otherwise they are  
called \emph{homogeneous} observers. We model the separate observation matrices $\bB^{(j)}$ for each observer $j = 
1,2,\cdots,m$, each of size $q \times \kappa_j$ in the heterogeneous case, and $\bB^{(j)} := \bB$ for all $j$ in the homogeneous 
case. In either setting, the HMM is described by the parameter set $\blambdam := \{\bpi; m, q, \{\kappa_j\}_{j=1}^m, \bA, \{\bB^{(j)}
\}_{j=1}^m\}$, with appropriate modifications to the observation matrix depending on the setting. For the homogeneous setting, 
we let $\kappa_j := \kappa$ and $\bB^{(j)} := \bB$, $\forall j$.

As we shall see, whether the observers are heterogenous or homogenous makes a significant difference to the 
identifiability of the HMM.

\subsection{Multi--letter tensor}

Unlike the single observer scenario, there are multiple observations in a single time step. We first consider the heterogeneous 
case, where, without loss of generality, we define the multi-letter tensor $\bcM_{\ast}$, an order $m+2$ tensor of dimensions $q 
\times q \times \kappa_1 \times \cdots \times \kappa_m$, with component matrices $\lbrack \bA, \id{q}, \{\bB^{(j)}\}_{j=1}^m \rbrack
$. 

Just as in the case of a single observer, the multi-letter tensor is connected to the HMM via its matricisation. Let $\kappa' := 
\prod_{j=1}^m \kappa_j$. Thus, the joint probability of a particular observed sequence $\by_1,\by_2,\cdots,\by_N$, noting each 
observation is now a vector, may be described by
\begin{align}
\nonumber
&P_{\blambda} (Y_1=\by_1, Y_2= \by_2,\cdots, Y_N = \by_N) \\
\label{eq:multi_sequence_prob}
&\hspace{1cm} = \bpi  \bW_{\ast} \bE(\by_1)  \bW_{\ast} \bE(\by_2) \cdots \bW_{\ast}  \bE(\by_N) \ones{q},
\end{align}
where $\bE(k)$ is a $\kappa' q \times q$ matrix, divided to $\kappa'$ row partitions, with the $q \times q$ identity matrix in 
the $k$-th row partition, $\ones{q}$ is the $q \times 1$ vector of ones, and 
\begin{align}
\nonumber
\bW_{\ast} &:= \bigorow{j=1}{m} \bB^{(j)} \orow \bA\\
\label{eq:hetero_mode}
&\,= \bB^{(1)} \orow \bB^{(2)} \orow \cdots \orow \bB^{(m)} \orow \bA.
\end{align}
In this formulation, all possible output sequences from space $\cY^{(1)} \times \cY^{(2)} \times \cdots \times \cY^{(m)}$ are 
ordered lexicographically, with $\bE(\by)$ selecting the correct position of $\by$ out of this ordering.
Note that $\bW_{\ast}$ is equivalent to $\bcM_{\ast}$ since it is the second mode matricisation of the tensor. 
Similarly, in the homogeneous setting, let $\bW_{\circ}$ have the same structure as \eqref{eq:hetero_mode}, with $\bB^
{(j)}:= \bB$, $\forall j$, and $\kappa' = \kappa^m$. 
%
%
We have the following result for $\bcM_{\ast}$.
\smallskip
\begin{lem}
Assume $m \ge 2$. The heterogeneous multi-letter tensor $\bcM_{\ast}$ is unique up to permutation and scaling of rows if and 
only if $\krank(\bigorow{j=1}{m} \bB^{(j)} \orow \bA) = q$. 
\label{lem:multi_observer_letter_hetero}
\end{lem}
\begin{proof}
The proof is similar to the proof of Lemma \ref{lem:single_letter}, extended to the multidimensional case, thus, we only need 
to sketch the proof here. We claim that it is necessary and sufficient that none of the non-trivial linear combinations of 
rows of $\bigorow{j=1}{m} \bB^{(j)} \orow \bA$ is expressed by a tensor product of two row vectors, 
i.e.~$\krank(\bigorow{j=1}{m} \bB^{(j)} \orow \bA) = q$. The chief ingredients are, (1) show a counterexample to proof
necessity, where the same example from the proof of Lemma \ref{lem:single_letter} can be used, appropriately modified to
account for additional dimensions, to construct an alternative decomposition of $\bcM_{\ast}$, and (2) for sufficiency, show
that $\id{q}$, the full rank component of $\bcM_{\ast}$ satisfies Kruskal's permutation lemma. Then, the claim is established.
\end{proof}

We must, however, be careful when the observers are independent and homogeneous. 
In this scenario, the above result no longer holds. Instead, the homogeneous setting is equivalent to the single observer
setting, evidenced by the following result.

\begin{lem}
It is necessary and sufficient that $\krank(\bB \orow \bA) = q$ for the homogeneous multi-letter tensor $\bcM_{\circ}$ to be 
unique up to permutation and scaling of rows.
\label{lem:multi_observer_letter_homo}
\end{lem}

An intuitive explanation is that each additional component, $\bB$, is exactly the same and do not provide sufficient variability 
for unique decomposition. Thus, even if the tensor $\bcM_{\circ}$ is matricised in different ways, one can always find an 
alternative decomposition of $\bcM_{\circ}$, similar to an example by Stegeman et al.~\cite{Stegeman05Perm} for the 3-way 
tensor. It is for this reason, decomposition-wise, $\bcM_{\circ}$ is no different from 
the single letter tensor $\bcL$.

\subsection{Identifiability conditions}

\begin{thm}
Suppose the $m$ observers are independent and homogeneous. For an HMM parameterised by $\blambdam$ to be unique up 
to a scaling and permutation of states, it is necessary and sufficient that the per letter tensor of the HMM satisfies Lemma 
\ref{lem:multi_observer_letter_homo}, i.e.~$\krank(\bB \orow \bA) = q$.
\label{thm:identifiability_multi_homo}
\end{thm}
\begin{proof}
The proof follows from the properties of $\bcM_{\circ}$. Then, it is clear $\krank(\bB \orow \bA) = q$ if and only if $\krank(\bigorow
{{}}{m} \bB \orow \bA) = q$, from Lemma \ref{lem:multi_observer_letter_homo}, otherwise an equivalent HMM can be 
constructed, such that the original HMM is no longer permutation and scaling indeterminate.
\end{proof}

The result shows that the homogeneous setting is essentially equivalent to the single observer setting.
As mentioned in \cite{Nguyen12MultiHMM}, if an HMM is unidentifiable in the single observer case, it is also
unidentifiable in the multiple independent homogenous observer case, as there is not enough variability in $\bcM_{\circ}$. 
Thus, no matter how many independent homogeneous observers are present, if the model cannot 
be identified in the single observer setting, then the model remains unidentifiable.


We next turn our attention to the heterogeneous case, a consequence of Lemma \ref{lem:multi_observer_letter_hetero}.
\begin{thm}
Suppose the $m$ observers are independent and heterogenous. For an HMM parameterised by $\blambdam$ to be 
unique up to a scaling and permutation of states, it is necessary and sufficient that the multi-letter tensor of the HMM satisfies 
Lemma \ref{lem:multi_observer_letter_hetero}, i.e.~$\krank(\bigorow{j=1}{m} \bB^{(j)} \orow \bA) = q$.
\label{thm:identifiability_multi_hetero}
\end{thm}

The extra dimensions of $\bcM_{\ast}$ are related to the additional advantage of having multiple independent observers. Each single 
observer $j = 1,2, \cdots, m$ is essentially restricted to a per letter tensor $\bcL^{(j)}$ which may not satisfy Lemma 
\ref{lem:single_letter} individually, but satisfies Lemma \ref{lem:multi_observer_letter_hetero} when their observations are 
jointly considered. 
This proves a natural advantage multiple 
independent heterogeneous observers have over a single observer, used to great effect in \cite{Nguyen12MultiHMM}.

%% file: conclusion.tex
\section{Conclusion}
\label{sec:conclusion}

In this paper, we revisit the identifiability of hidden Markov models, via tensor decomposition, 
where there are well-known results regarding the permutation and scaling indeterminacy of tensors. We proved 
deterministic identifiability conditions  of single and multi-observer HMMs using well--established results on the Kruskal 
rank. Our results are stronger than previous results, where only generic identifiability based on the regularity of the HMM is 
assumed. Future work includes using our framework to provide insights in the inference of the transition and 
observation matrices of HMMs and extend the work to dependent observers.